\documentclass{article}
\usepackage{spconf}
\usepackage{color}
\usepackage{epsfig}
\usepackage{amsmath}
\usepackage{amsfonts}
\usepackage{amsthm}
\usepackage{amssymb}
\usepackage{accents}
\usepackage{graphics}
\usepackage{graphicx}
\usepackage{float}
\usepackage{flushend}

\newtheorem*{example*}{Example}

\newtheorem*{definition*}{Definition}

\newtheorem*{theorem*}{Theorem}

\newtheorem*{lemma*}{Lemma}

\newtheorem{algorithm}{Algorithm}
\let\olddefinition\algorithm
\renewcommand{\algorithm}{\olddefinition\normalfont}

\newcommand{\PPB}{\mathcal{P}(\mathcal{B})}
\newcommand{\PPBbar}{\overline{\mathcal{P}}(\mathcal{B})}
\newcommand{\PPBo}{\accentset{\circ}{\mathcal{P}}(\mathcal{B})}
\newcommand{\CB}{\mathcal{C}(\mathcal{B})}

\newcommand{\Z}{\mathbb{Z}}
\newcommand{\R}{\mathbb{R}}

\newcommand{\F}{\mathbb{F}}
\newcommand{\hx}{\hat{x}}
\newcommand{\hz}{\hat{z}}

\DeclareMathOperator*{\vol}{Vol}

\title{Neural Lattice Decoders}

\name{Vincent Corlay$^{\dagger,*}$, Joseph J. Boutros$^{\ddagger}$, Philippe Ciblat$^{\dagger}$, and Lo\"ic Brunel$^*$}
\address{$^{\dagger}$ Telecom ParisTech, 46 Rue Barrault, 75013 Paris, v.corlay@fr.merce.mee.com\\
$^{\ddagger}$ Texas A\&M University, Doha, Qatar, $^*$Mitsubishi Electric R\&D, Rennes, France.}

\begin{document}

\maketitle

\begin{abstract}
Lattice decoders constructed with neural networks are presented. Firstly, we show how the fundamental parallelotope
is used as a compact set for the approximation by a neural lattice decoder. Secondly, we introduce
the notion of Voronoi-reduced lattice basis. As a consequence, a first optimal neural lattice decoder is built
from Boolean equations and the facets of the Voronoi cell. This decoder needs no learning.
Finally, we present two neural decoders with learning. It is shown that L1 regularization and {\em a priori} information
about the lattice structure lead to a simplification of the model.
\end{abstract}


\section{Neural Decoding via a Compact Set \label{sec_compact}}
We restrict this paper to point lattices in the $n$-dimensional real space $\R^n$,
also called Euclidean lattices.
A lattice $\Lambda$ is a free $\Z$-module in $\R^n$, or simply a discrete additive subgroup of $\R^n$.
To generate such an infinite discrete set with an additive group structure,
$\Lambda$ requires a basis formed by linearly independent vectors.
For a rank-$n$ lattice in $R^n$, the rows of a $n\times n$ generator matrix $G$ constitute
the basis of $\Lambda$ and any lattice point $x$ is obtained via $x=zG$, where $z \in \Z^n$.
For a given basis $\mathcal{B}=\{ g_i \}_{i=1}^n$ forming the rows of $G$,
the fundamental parallelotope of $\Lambda$ is defined by
\begin{equation}
\label{equ_PB}
\mathcal{P}(\mathcal{B}) = \{ y \in \R^n : y=\sum_{i=1}^{n}\alpha_{i}g_{i}, \ 0 \leq \alpha_{i} < 1  \}.
\end{equation}
For $\PPB$ in (\ref{equ_PB}), we also define its closure denoted by $\PPBbar$.
The fundamental volume of $\Lambda$ is $\det(\Lambda)=|\det(G)|=\vol(\mathcal{V}(x))=\vol(\mathcal{P}(\mathcal{B}))$.
The Voronoi cell of $x$ is:
\begin{equation}
\mathcal{V}(x)=\{ y \in \R^n : \|y-x\| \le \|y-x'\|, \forall x' \in \Lambda \}.
\end{equation}
A vector $v \in \Lambda$ is called Voronoi vector if the half-space $\{y \in \mathbb{R}^{n} \ : \ y.v \leq \frac{1}{2}v.v \}$ has a non empty intersection with $\mathcal{V}(0)$. 
The vector is said relevant if the intersection is an $(n-1)$-dimensional face of $\mathcal{V}(0)$.
The Voronoi cell is thus also defined as the intersect of half-spaces:
\begin{equation}
\mathcal{V}(0)=\{y \in \mathbb{R}^{n} : y.v \leq \frac{1}{2}v.v \},
\end{equation}
where $v$ is relevant. We call the number of relevant Voronoi vector the Voronoi number. 
For root lattices \cite{Conway1999} the Voronoi number is equal to the kissing number.
The first minimum of $\Lambda$, i.e. its minimum Euclidean distance, is given by
\begin{equation}
d_{min}(\Lambda)=\min_{x \ne x'} \| x-x'\|=2\rho,
\end{equation}
for $x,x' \in \Lambda$ and where $\rho$ is the packing radius of the associated lattice sphere packing.

Lattice decoding refers to the method of finding the closest lattice point, the closest in Euclidean distance sense.
This problem is also known as the Closest Vector Problem (CVP).
Its associated decision problem is NP-complete \cite[Chap.~3]{Micciancio2002}. 
Nevertheless, lattice decoding has been extensively studied in the literature for small and large dimensions.
Optimal and quasi-optimal decoders are known for random lattices encountered in communications channels,
typically for $n \le 100$ \cite{Viterbo1999}\cite{Agrell2002}\cite{Brunel2003},
and for binary Construction-A lattices \cite[Chap.~20]{Conway1999}.
Sub-optimal message-passing iterative decoders were very successful for non-binary Construction-A lattices 
in dimensions as high as 1 million \cite{Boutros2015}\cite{diPietro2017}.
The literature also includes extensive work on decoding multi-level coded modulations \cite{Wachsmann1999}
that give rise to lattice (coset codes) and non-lattice constellations,
e.g. see \cite{Vardy1993}\cite{Sadeghi2006}\cite{Sommer2008}\cite{Yan2014}
for Leech and lattices based on low-density parity-check codes and polar codes.\\

\noindent
$\bullet$ {\bf Relation to Prior Work}. This is the first paper describing how artificial neural networks
can be employed to solve the CVP for infinite lattice constellations. We build an optimal neural lattice decoder
based on Voronoi-reduced lattice bases. This first neural lattice decoder needs no learning and is tractable in small dimensions.
We also show two other types of neural lattice decoders obtained from training of unconstrained and constrained feed-forward networks.
As cited above, previous published lattice decoders do not utilize neural networks techniques.
The current literature on machine learning and deep learning includes interesting results with applications
to communication theory and coding theory, e.g. \cite{Nachmani2016}\cite{Samuel2017}\cite{Kim2018}.
These results motivated us to develop neural lattice decoders.
After we submitted the first version of this manuscript, two other papers
were also published on machine learning for lattice decoding \cite{Mohammadkarimi2018} \cite{Sadeghi2018}.\\

Neural network classifiers are trained to take the best decision about the value of a variable that belongs to a finite set,
most frequently the binary set $\F_2$ \cite{Goodfellow2016}.
Other learning models are trained to produce a good estimate for a real number
characterizing one variable, e.g. the probability of a given event involving that variable.
In other words, to our modest knowledge, it appears that feed-forward networks
do not have the capability of observing the entire space $\R^n$ to infer the value of a variable that belongs to an infinite set
such as $\Z$.
Indeed, the majority of known lattice decoders search for the closest lattice point $\hx=\hz G$ by looking in $\Z^n$ for the best
vector $\hz$, except for low-density lattices where message passing solves directly the coordinates of $\hx$ \cite{Sommer2008}.
In order to help a neural network solve or approximate the CVP, we force the decoder input to satisfy
the assumptions of the Universal Approximation Theorem. This theorem
was proved by G.~Cybenko for sigmoid networks \cite{Cybenko1989}
and then generalized by K.~Hornik to multilayer feed-forward architectures \cite{Hornik1991}.
A version of this theorem can be stated as follows \cite{Anthony1999}:

\begin{theorem*}
{\bf (Anthony \& Bartlett 1999)}. The two-layer sigmoid networks are ``universal approximators", in a sense that,
given any continuous function $f$ defined on some compact subset $\mathcal{S}$ of $\R^n$,
and any desired accuracy $\epsilon$, there is a two-layer sigmoid network computing a function
that is within $\epsilon$ of $f$ at each point of $\mathcal{S}$.
\end{theorem*}

We focus on the fact that $f$ is defined on a compact subset of $\R^n$.
There are many ways to partition $\R^n$. Two obvious partitions inspired from the lattice structure are
$\R^n=\bigcup_{x \in \Lambda} \mathcal{V}(x)$, where one should be careful in assigning the facets to a single Voronoi cell,
and $\R^n=\bigcup_{x \in \Lambda} (\mathcal{P}(\mathcal{B})+x)$.
Note that $\mathcal{V}(x)=\mathcal{V}(0)+x$. 
The partition based on Voronoi cells cannot be used because, given $y \in \R^n$, solving $\hx$ where $y \in \mathcal{V}(\hx)$
is exactly the CVP that we aim to solve. On the other hand, it is easy to determine the translated parallelotope
$\mathcal{P}(\mathcal{B})+x$ to which $y$ belongs. Hence, the lattice decoder input $y$ is translated by $-x$
to let the neural lattice decoder operate in the compact region $\PPBbar+0$ or equivalently in $\PPB+0$
by assigning half of the facets (the upper facets) to the neighboring parallelotope.

\section{Voronoi-Reduced Lattice Basis \label{sec_voronoi}}
In the sequel, following the conclusion of the previous section,
our neural lattice decoder shall operate as follows within the fundamental parallelotope (Step 2 below):
\begin{itemize}
\itemsep=-1mm
\item Step 0: A noisy lattice point $y_0=x+\eta$ is observed, where  $x \in \Lambda$ and $\eta \in \R^n$ is an additive noise.
\item Step 1: Compute $t=\lfloor y_0G^{-1} \rfloor$ and get $y=y_0-tG$ which now belongs to $\mathcal{P}(\mathcal{B})$.
Note: the floor function applied to a vector corresponds to its application on all its coordinates.
\item Step 2: The neural lattice decoder finds $\hx$, the closest lattice point to $y$.
\item Step 3: The closest point to $y_0$ is $\hx_0=\hx+tG$. 
\end{itemize}

For mod-2 Construction-A lattices, the noisy point $y_0$ can be folded inside the cube $[-1,+1]^n$ to decode the component
error-correcting code and then find the closest lattice point~\cite[Chap.~20]{Conway1999}.
Thus, another option for the compact set to be used by the neural decoder of mod-2 Construction-A lattices
is the cube $[-1,+1]^n$. In this paper, although the HLD decoder described in Section~\ref{secMERDEHLD}
can also operate on $[-1,+1]^n$, we will only consider the compact region $\PPBbar$ for Step~2.

\begin{figure}[!t]
\centering
\includegraphics[scale=0.6]{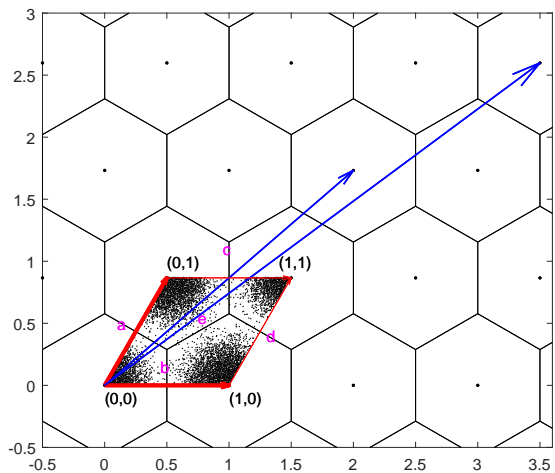}
\caption{Voronoi-reduced basis $\mathcal{B}_1$ for $A_2$ (in red) and a non-reduced basis $\mathcal{B}_2$ (in blue).
$\mathcal{P}(\mathcal{B}_1)$ is partitioned into 4 parts included in the Voronoi cells of its corners.
$\mathcal{P}(\mathcal{B}_2)$ has 10 parts involving 10 Voronoi cells.\label{fig_A2}}
\end{figure}

\begin{definition*}
\label{def_Voronoi-reduced}
Let $\mathcal{B}$ be the $\Z$-basis of a rank-$n$ lattice $\Lambda$ in~$\R^n$.
$\mathcal{B}$ is said Voronoi-reduced if, for any point $y \in \mathcal{P}(\mathcal{B})$,
the closest lattice point $\hx$ to $y$ is one of the $2^n$ corners of $\mathcal{P}(\mathcal{B})$,
i.e. $\hx=\hz G$ where $\hz \in \{0, 1\}^n$.
\end{definition*}

We will use the abbreviation {\em VR basis} to refer to a Voronoi-reduced basis.
Figure~\ref{fig_A2} shows the hexagonal lattice $A_2$, its Voronoi cells, and the fundamental parallelotope of the basis
$\mathcal{B}_1=\{v_1, v_2 \}$, where $v_1=(1, 0)$ corresponds to $z=(1, 0)$
and $v_2=(\frac{1}{2}, \frac{\sqrt{3}}{2})$ corresponds to $z=(0, 1)$.
The basis $\mathcal{B}_1$ is Voronoi-reduced because
\[
\mathcal{P}(\mathcal{B}_1) \subset \mathcal{V}(0) \cup \mathcal{V}(v_1) \cup \mathcal{V}(v_2) \cup\mathcal{V}(v_1+v_2).
\]

Lattice basis reduction is an important field in Number Theory.
We cite three famous types of reduction to get a good basis:
Minkowski-reduced basis, Korkin-Zolotarev-reduced (or Hermite-reduced) basis, and LLL-reduced basis
for Lenstra-Lenstra-Lov\'asz \cite{Micciancio2002}\cite{Cohen1996}.
The reader may notice that a basis with all its vectors on the first lattice shell is Minkowski-reduced.
In general, non-dense lattices do not admit a basis from the first shell.
The basis $\mathcal{B}_1$ in Figure~\ref{fig_A2} is Minkowski-, KZ-, and Voronoi-reduced.
The basis $\{v_1, v_1+v_2\}$ of $A_2$ is not Minkowski, however it is Voronoi-reduced. 
Famous densest lattices listed in \cite{Conway1999}, $D_4$, $E_8$, $\Lambda_{16}$, and $\Lambda_{24}$,
all have a basis from their first shell, however the VR property is not always guaranteed.
Currently, we are completing the study of properties
and existence of a VR basis for a given Euclidean lattice.

The 3-dimensional lattice $A_3=D_3$ and the 4-dimensional Schl\"afli lattice $D_4$ both admit a VR basis.
A VR basis, when it exists, is not necessarily unique. 
The Gram matrices $\Gamma_3$ and $\Gamma_4$ of a VR basis for $A_3$ and $D_4$ are:
\begin{equation}
\Gamma_3=
\left(
\begin{array}{ccc}
2 & 1 & 0 \\
1 & 2 & 1 \\
0 & 1 & 2
\end{array}  
\right),
\end{equation}

\begin{equation}
\Gamma_4=
\left(
\begin{array}{cccc}
  2 & 1 & 1 & 1 \\
  1 & 2 & 1 & 1 \\
  1 & 1 & 2 & 0 \\
  1 & 1 & 0 & 2
\end{array}  
\right).
\end{equation}
Recall that the Gram matrix is $\Gamma=GG^t=(GQ)(GQ)^t$ \cite{Conway1999}, where $Q$ is any $n \times n$ orthogonal matrix.
A lower triangular generator matrix is obtained from the Gram matrix by Cholesky decomposition. 
The Gosset lattice $E_8$ admits a VR basis with respect to $\PPBo$, i.e. there exist isolated points on the facets of $\PPB$
that are not decoded to its corners. A Voronoi-reduced basis of $E_8$ is given by the following Gram matrix:
\begin{equation}
\Gamma_8=
\left(
\begin{array}{cccccccc}
  4  & 2  & 0  & 2  & 2  & 2  & 2  & 2 \\
  2  & 4  & 2  & 0  & 2  & 2  & 2  & 2 \\
  0  & 2  & 4  & 0  & 2  & 2  & 0  & 0 \\
  2  & 0  & 0  & 4  & 2  & 2  & 0  & 0 \\
  2  & 2  & 2  & 2  & 4  & 2  & 2  & 0 \\
  2  & 2  & 2  & 2  & 2  & 4  & 0  & 2 \\
  2  & 2  & 0  & 0  & 2  & 0  & 4  & 0 \\
  2  & 2  & 0  & 0  & 0  & 2  & 0  & 4
\end{array}  
\right).
\end{equation}

For intermediate dimensions, e.g. $n=6$, and for higher dimensions, e.g. $n=16$ and $n=24$,
when the existence of a VR basis cannot be proved via algebraic tools or via a tractable computer search,
the strong constraint defining a VR basis can be relaxed.

\begin{definition*}
Let $\mathcal{C}(\mathcal{B})$ be the set of the $2^n$ corners of $\PPB$. Let $O$ be the subset of $\PPB$
that is covered by Voronoi cells of points not belonging to $\CB$, namely
\begin{equation}
O=\PPB \setminus \left(\PPB \bigcap \left(\bigcup_{x \in \CB} V(x)\right) \right).
\end{equation}
$\mathcal{B}$ is said quasi-Voronoi-reduced if $\vol(O) \ll \det(\Lambda)$. 
\end{definition*}
Let $d^2_{OC}(\mathcal{B})=\min_{x \in O, x' \in \CB} \|x-x'\|^2$ be the minimum squared Euclidean distance between $O$ and $\CB$.
The sphere packing structure associated to $\Lambda$ guarantees that $d^2_{OC} \ge \rho^2$. 
Let $Pe(\mathcal{B})$ be the probability of error (per point) for a neural lattice decoder
built from  a quasi-Voronoi-reduced basis $\mathcal{B}$. Here, we assume that $y_0=x+\eta$
with $\eta_i \sim \mathcal{N}(0,\sigma^2)$, for $i=1, \ldots, n$. The following lemma tells us
that a quasi-Voronoi-reduced basis exhibits quasi-optimal performance on a Gaussian channel
at high signal-to-noise ratio.
In practice, the quasi-optimal performance is also observed at moderate values of signal-to-noise ratio. 
\begin{lemma*}
\label{lem_quasi-Voronoi-reduced}
\begin{align}
  Pe(\mathcal{B}) \le & ~\frac{\tau}{2} \exp(-\frac{\pi e \Delta\gamma}{4})
  + o\left(\exp(-\frac{\pi e \Delta\gamma}{4})\right)\label{equ_Peopt}\\
  & + \frac{\vol(O)}{\det(\Lambda)} \cdot (e\Delta)^{n/2} \cdot
  \exp(-\frac{\pi e \Delta\gamma}{4} \cdot \frac{d^2_{OC}}{\rho^2}),\label{equ_PeO}
\end{align}
for $\Delta$ large enough,
where $\Delta=\frac{\det(\Lambda)^{2/n}}{2\pi e \sigma^2}$ is the distance to Poltyrev limit \cite{Poltyrev1994}, 
$\gamma$ is the Hermite constant of $\Lambda$ \cite{Conway1999},
and $o()$ is the small o Bachmann-Landau notation.
\end{lemma*}
\begin{proof}
For a complete maximum-likelihood decoder (optimal) on the Gaussian channel, the probability of error
per lattice point can be bounded from above by
\begin{equation}
P_e(opt) \le \frac{1}{2}\Theta_{\Lambda}\left(q=\exp(-\frac{1}{8\sigma^2})\right) - \frac{1}{2},
\end{equation}
where $\Theta_{\Lambda}(z)=\sum_{x \in \Lambda} q^{\|x\|^2}$ is the Theta series of $\Lambda$,
see (35) in Section~1.4 of Chapter~3 in \cite{Conway1999}.
It can be easily shown that $\frac{\rho^2}{2\sigma^2}=\frac{\pi e \Delta\gamma}{4}$.
For $\Delta \rightarrow \infty$, the term $\tau q^{4\rho^2}$ dominates the sum in $\Theta_{\Lambda}(z)$, then
\[
P_e(opt) \le ~\frac{\tau}{2} \exp(-\frac{\pi e \Delta\gamma}{4})
  + o\left(\exp(-\frac{\pi e \Delta\gamma}{4})\right).
\]
If $\mathcal{B}$ is Voronoi-reduced and the neural lattice decoder works inside $\PPB$
to find the nearest corner, then the performance is given by $P_e(opt)$.\\
If $\mathcal{B}$ is quasi-Voronoi-reduced and the neural decoder only decides a lattice point from $\CB$,
then an error shall occur each time $y$ falls in $O$. We get
\begin{equation}
P_e(\mathcal{B}) \le P_e(opt) + P_e(O),
\end{equation}
where
\begin{align*}
P_e(O) & =\idotsint_O \frac{1}{\sqrt{2\pi \sigma^2}^n} \exp(-\frac{\|x\|^2}{2\sigma^2}) \,dx_1 \dots dx_n\\
& \le \frac{1}{\sqrt{2\pi \sigma^2}^n} \exp(-\frac{d^2_{OC}}{2\sigma^2})~\vol(O)\\
& = \frac{\vol(O)}{\det(\Lambda)} \cdot (e\Delta)^{n/2} \cdot
  \exp(-\frac{\pi e \Delta\gamma}{4} \cdot \frac{d^2_{OC}}{\rho^2}).
\end{align*}
This completes the proof.
\end{proof}
\noindent
The following Gram matrix corresponds to a quasi-Voronoi-reduced basis of $E_6$,
\begin{equation}
\Gamma_6=\left(
\begin{array}{cccccc}
  3 & \frac{3}{2} & 0 & 0 &\frac{3}{2} &\frac{3}{2} \\
  \frac{3}{2} & 3 & 0 & 0 &\frac{3}{2} &\frac{3}{2} \\  
  0 & 0 & 3 & \frac{3}{2} &\frac{3}{2} &\frac{3}{2} \\
  0 & 0 &\frac{3}{2} & 3 &\frac{3}{2}&\frac{3}{2} \\
  \frac{3}{2} & \frac{3}{2} &\frac{3}{2} &\frac{3}{2} & 3 & \frac{3}{2} \\
  \frac{3}{2} & \frac{3}{2} & \frac{3}{2} &\frac{3}{2} &\frac{3}{2} & 3
\end{array}
\right),
\end{equation}
with $\frac{d^2_{OC}}{\rho^2}=1.60$ (2dB of gain) and $\frac{\vol(O)}{\det(\Lambda)}=2.47\times 10^{-3}$.
The ratio of (\ref{equ_PeO}) by (\ref{equ_Peopt}) is about $10^{-4}$ at $\Delta=1=0dB$ (on top of Poltyrev limit!)
then vanishes further for increasing $\Delta$. 
Obviously, the quasi-VR property is good enough to allow the application
of a neural lattice decoder working with $\CB$ such as the Hyperplane
Logical Decoder presented in the next section.
If a complete decoder is required, e.g. in specific applications such as lattice shaping and cryptography,
the user should let the neural lattice decoder manage extra points outside $\CB$. For example,
the disconnected region $O$ for $E_6$ defined by $\Gamma_6$ includes extra points
where $z_i \in \{ -1, 0, 1, +2\}$ instead of $\{ 0, 1\}$ as for $\CB$. 

\section{A Hyperplane Logical Decoder}\label{secMERDEHLD}
In this section, we introduce a neural lattice decoder to find the closest point
for small dimensions without learning. This decoder, referred to as the Hyperplane Logical Decoder (HLD),
is Maximum-Likelihood (it exactly solves the CVP) for lattices admitting a VR basis.
It can also be applied to lattices admitting only a quasi-VR basis,
to yield near-Maximum-Likelihood performance in presence of additive white Gaussian noise. 

The HLD shall operate in $\mathcal{P}=\mathcal{P}(\mathcal{B})$ as for Step~2 in the decoding steps
listed in the previous section.
$\mathcal{B}$ is assumed to be Voronoi-reduced. The exact CVP, or Maximum-Likelihood Decoding (MLD),
is solved by comparing the position of $y$ to all Voronoi facets partitioning $\mathcal{P}$. 
This can be expressed in the form of a Boolean equation,
where the binary (Boolean) variables are the positions with respect to the facets (on one side or another).
Since $\mathcal{V}(x)=\mathcal{V}(0)+x$, orthogonal vectors to all facets partitioning $\mathcal{P}$
are determined from the facets of $\mathcal{V}(0)$.

\begin{example*}
On Figure~\ref{fig_A2}, let $\hz=(\hz_1, \hz_2)$. The first component $\hz_1$ is $1$ (true) if
$y$ is on the right of hyperplane $c$, or on the right of $b$ and below $e$ simultaneously. 
As slight abuse of notation, we let $a,b,c,d$ and $e$ be Boolean variables, the state of which depends on the location of $y$ with respect to the corresponding hyperplane. 
We get the Boolean equation $\hz_1=c+b \cdot e$, where $+$ is a logical OR and $\cdot$ stands for a logical AND.
Similarly, $\hz_2=d+a \cdot \overline{e}$, where $\overline{e}$ is the Boolean complement of $e$.
\end{example*}

\begin{figure}[!ht]
\centering
\includegraphics[scale=0.835]{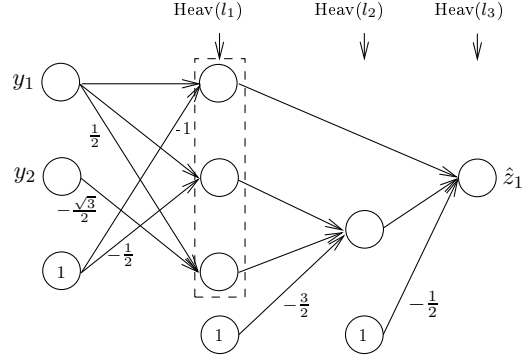}
\caption{Neural network performing HLD decoding on the first symbol $z_{1}$ of a point in $\mathcal{P}$ for the lattice
$A_2$. Unlabeled edges have weight 1. The bias nodes are required to perform AND and OR.
Heav($\cdot$) stands for Heaviside($\cdot$).}
\label{fig_A2_net}
\end{figure}

For $\Lambda \subset \R^n$ of rank $n$, to find the Boolean equation of a coordinate $\hz_k$,
select the $2^{n-1}$ corners of $\mathcal{P}$ where $z_{k}=1$ and perform the two following steps:
\begin{itemize}
\item For each corner, move in the direction of a relevant vector by half its norm + $\epsilon$ (e.g. by $\rho + \epsilon$ if the relevant vector is a lattice point from the first shell).
Three possible situations are encountered. (i) The resulting point is outside $\mathcal{P}$.
Hence, there is no decision boundary in this direction. (ii) If not outside $\mathcal{P}$,
find the closest lattice point $x'=z'G$ by sphere decoding \cite{Viterbo1999}\cite{Agrell2002}.
If $z'_{k}=1$ then, again, there exists no decision boundary in this direction.
(iii) $z'_{k}=0$, a decision boundary orthogonal to this direction does exist. 
\item The Boolean equation of $\hz_k$ contains a term with a Boolean AND of all decision boundaries found
at the same corner. The equation is the Boolean OR of $2^{n-1}$ terms coming from all selected corners.
\end{itemize}

For clarity reasons, we omitted technical details in the above steps that involve facets of $\mathcal{P}$ with a tie.
In practice, the constructed Boolean equation with its $2^{n-1}$ terms is significantly reduced
into a simpler equation, mainly as a result of identical terms.
If $\Lambda$ does not admit a VR or quasi-VR basis, the HLD should be constructed from more lattice shells and some coordinates of $\hz$ are not binary anymore.

Once the Boolean equations and the decision boundaries are known,
the HLD can be executed in three steps (a)-(c). By abuse of terminology,
the inner product of two points in $\R^n$ refers to the inner product
between the two vectors defined by these points: 
(a) Compute the inner product of $y$ with the lattice points orthogonal to the decision boundaries.
(b) Apply the Heaviside function on the resulting quantities to get its relative position
under the form of Boolean variables.
(c) Compute the logical equations associated to each coordinate.

Since (a) and (b) are simply inner products followed by activation functions,
the natural way to represent these steps is to use perceptrons \cite{Goodfellow2016},
where the edges are labeled with the decision hyperplane parameters,
i.e. the perceptron weights define the vector orthogonal to the decision hyperplane.
(a) and (b) form the first layer of a neural network.
The second layer implements the logical AND and the third layer the logical OR.
As a result, the HLD can be thought of as a neural network with two hidden layers.
Figure~\ref{fig_A2_net} illustrates the topology of the neural network obtained
when applying the HLD to the lattice $A_2$. 

Figure~\ref{fig_HLD_lat_perf} shows the point error-rate performance of MLD and HLD for the Schl\"afli lattice $D_4$
and the Gosset lattice $E_8$. All decoders perform exact CVP. However, HLD was running on a distinct machine
with a different pseudo-random sequence of lattice points and noise samples. This explains the slight difference
between HLD and MLD in Figure~\ref{fig_HLD_lat_perf} due to the Monte Carlo method.

\begin{figure}[!t]
\centering
\includegraphics[scale=0.335,angle=-90]{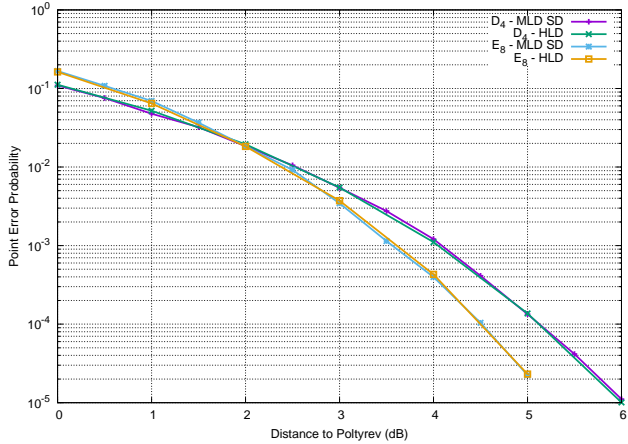}
\caption{HLD and MLD for lattices $D_4$ and $E_8$.}
\label{fig_HLD_lat_perf}
\end{figure}

\section{Learning to decode}
We discuss here two neural lattice decoders denoted by NLD2 and NLD3. The first neural lattice decoder (NLD) is the HLD
of the previous section. Both models NLD2 and NLD3 need to acquire their weights via learning.

NLD2 is a standard fully-connected feed-forward sigmoid network \cite{Goodfellow2016}
without any constraint on its architecture. To be competitive, the number of parameters of NLD2 should grow slower than $2^{n}$.
The discussion of the sample complexity \cite{Anthony1999} of NLD2 is omitted due to lack of space. 
The performance of NLD2 is shown on Figure~\ref{fig_NLD2} for $E_8$ ($n=8$) and the MIMO lattice $T55$ ($n=16$) taken
from \cite{Samuel2017}. In all NLD2 models, the size of first hidden layer is taken to be of the same order of magnitude as
the lattice kissing number ($\tau(E_8)=240$ and $\tau(T55)=30$). For $E_8$, the NLD2 has three hidden layers each with 200 neurons. 
Its performance is very close to MLD but this model has $W=83200$ parameters and is too complex relative to HLD for $E_8$.
The ratio $\frac{\text{log}_{2}(W)}{n}=2.0$ (supra linear).
The NLD2 neural network is not suited to decoding dense lattices.
For $T55$, the NLD2 in case~1 has three hidden layers with 50-100-100 neurons respectively. In case~2, it is also 
made up of three hidden layers with 30-50-50 neurons respectively and $W=6280$ parameters.
The ratio $\frac{\text{log}_{2}(W)}{n}=0.78$ (sub-linear).
From its complexity and its illustrated performance,
we state that NLD2 is a competitive decoding algorithm for non-dense lattices. 

\begin{figure}[t!]
\centering
\includegraphics[scale=0.335,angle=-90]{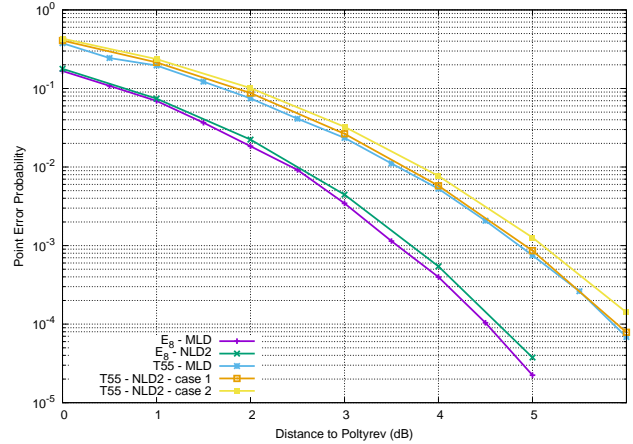}
\caption{Decoders with learning, NLD2, without constraint.}
\label{fig_NLD2}
\end{figure}

Now, we introduce NLD3, a learning model with L1 regularization to simplify its structure. 
NLD3 shall have a structure constraint: its first hidden layer is fixed and taken from the HLD model.
Indeed, the authors in \cite{Nachmani2016} used the neural network representation of the Tanner graph for BCH codes
to come up with an architecture exploiting \textit{a priori} information on the code structure.
The neural network being sub-optimal, they improved its performance via learning.
We embrace a similar paradigm: use \textit{a priori} information on the structure of the lattice
to build the architecture. Nevertheless, in our case, the HLD neural network is already optimal (it cannot be improved).
However, we let an HLD-initialized NLD3 model learn to simplify its structure while limiting the performance degradation. We provide an example with $D_4$.
For the first coordinate $z_{1}$, the HLD has the following Boolean equation: 
\vspace{-2.1mm}
\[
z_{1} =  u_{1}+ u_{2} \cdot u_{3} \cdot u_{4} \cdot u_{5} \cdot u_{6} + u_{4} \cdot u_{7} \cdot u_{8} 
          + u_{4} \cdot u_{7} \cdot u_{9}   +  u_{4}\cdot u_{10}. 
\vspace{-2.1mm}
\]
The performance of NLD3 with $D_4$ is given in Figure~\ref{fig_NLD3}.
Cases~1 and 2 correspond to the results of L1 regularization with respectively Heaviside and sigmoid activation functions. For cases~3 and 4, in order to further simplify the model structure, three edges were pruned with two different strategies between the first and the second hidden layers.
Of course, learning does not lead to the same model weights in these cases.
The NLD3 model in case~1 simplifies the Boolean equation of $z_1$ to two terms only (to be compared to the five AND conditions above):
the second hidden layer shrank from five to two neurons.
Its performance is still Maximum-Likelihood like the HLD.
Case~2 is also quasi-optimal, the slight loss is due to an imperfect training.
Case~3 generates an error-floor while case~4 exhibits a great robustness.
A similar behavior of NLD3 was observed when utilized on other point lattices.

\begin{figure}
\centering
\includegraphics[scale=0.33,angle=-90]{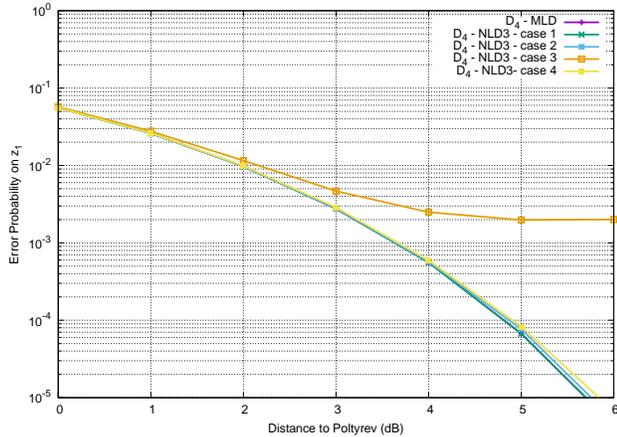}
\caption{Neural lattice decoder with learning, NLD3, with constraints and L1 regularization, applied to the lattice $D_4$.}
\label{fig_NLD3}
\end{figure}


\end{document}